\documentclass[11pt,reqno,a4paper]{amsart}
\usepackage[T1]{fontenc}
\usepackage[utf8]{inputenc}
\usepackage{amsmath}
\usepackage{algorithm}
\usepackage{amssymb}
\usepackage{graphicx}

\RequirePackage{enumitem}
\setenumerate{itemsep=0pt}
\setitemize{itemsep=0pt}

\usepackage{cases}
\usepackage{subfigure}
\usepackage{verbatim}
\usepackage{comment}
\usepackage{amsthm}
\usepackage{thmtools}
\usepackage{url}

\newcommand{\B}{\mathrm{B}}
\newcommand{\D}{\mathrm{D}}
\newcommand{\sca}[2]{\langle #1 | #2\rangle}
\newcommand{\Haus}{\mathrm{vol}}
\newcommand{\Lag}{\mathrm{Lag}}
\newcommand{\Class}{\mathcal{C}}
\newcommand{\Rsp}{\mathbb{R}}
\newcommand{\Nsp}{\mathbb{N}}

\usepackage[noend]{algpseudocode}
\usepackage{algorithmicx}

\usepackage{stmaryrd}

\usepackage{color}
\definecolor{red}{rgb}{1,0,0}

\DeclareMathOperator*{\argmax}{arg\,max}
\DeclareMathOperator{\Span}{span}
\DeclareMathOperator{\rk}{rk}
\DeclareMathOperator{\im}{im}
\DeclareMathOperator{\inter}{int}
\DeclareMathOperator\supp{spt}
\DeclareMathOperator\Pow{Pow}

\newcommand{\Tphi}{T_\phi}
\renewcommand{\Tphi}{S_\phi}

\newcommand{\R}{\mathbb{R}}
\newcommand{\N}{\mathbb{N}}
\newcommand{\1}{\mathrm{\textbf{1}}}
\newcommand{\OneN}{\llbracket 1,N \rrbracket}
\renewcommand{\D}{\mathrm{D}}
\renewcommand{\phi}{\varphi}
\renewcommand{\Haus}{\mathcal{H}}
\newcommand{\Sad}{\mathcal{S}^{\alpha,\delta}} 

\newtheorem{theorem}{Theorem}
\newtheorem{lemma}[theorem]{Lemma}
\newtheorem{proposition}[theorem]{Proposition}
\newtheorem{corollary}[theorem]{Corollary}

\theoremstyle{definition}
\newtheorem{definition}[theorem]{Definition}

\theoremstyle{remark}
\newtheorem{remark}[theorem]{Remark}

\title{A Damped Newton algorithm for Generated Jacobian Equations}
\author{Anatole Gallou\"{e}t}
\address{Univ. Grenoble Alpes, CNRS, Grenoble INP, LJK, 38000 Grenoble, France}
\email{anatole.gallouet@univ-grenoble-alpes.fr}
\author{Quentin M\'erigot}
\address{Université Paris-Saclay, CNRS, Laboratoire de mathématiques d’Orsay, 91405, Orsay, France \and
  Institut universitaire de France (IUF)}
\email{quentin.merigot@universite-paris-saclay.fr}
\author{Boris Thibert}
\address{Univ. Grenoble Alpes, CNRS, Grenoble INP, LJK, 38000 Grenoble, France}
\email{Boris.Thibert@univ-grenoble-alpes.fr}

\date{\today}
\begin{document}
\maketitle

\begin{abstract}
 Generated Jacobian Equations have been introduced by Trudinger
 [Disc. cont. dyn. sys (2014), pp. 1663--1681] as a generalization of
 Monge-Ampère equations arising in optimal transport. In this paper, we
 introduce and study a damped Newton algorithm for solving these equations in
 the \emph{semi-discrete} setting, meaning that one of the two
 measures involved in the problem is finitely supported and the other one is absolutely continuous. We also
 present a numerical application of this algorithm to the near-field
 parallel refractor problem arising in non-imaging problems.
\end{abstract}

\section{Introduction}
This paper is concerned with the numerical resolution of \emph{Generated Jacobian equations}, introduced by N. Trudinger~\cite{trudinger2012local} as a generalization of Monge-Ampère equations arising in optimal transport. Generated Jacobian equations  were originally motivated by inverse problems arising in non-imaging optics in the near-field case~\cite{kochengin1997determination,guillen2017pointwise,gutierrez2015regularity} but they also apply to problems arising in economy~\cite{noldeke2018implementation,galichon2019costly}. A survey on these equations and their applications was recently written by N. Guillen~\cite{guillen2019primer}.  The input for a generated Jacobian equations are two probability measures $\mu$ and $\nu$ over two spaces $X$ and $Y$, and a \emph{generating function} $G:X\times Y\times \Rsp\to\Rsp$. Loosely speaking, a scalar function $\psi$ on $Y$ is an Alexandrov solution to the generated jacobian equation if the map
$T_\psi$ defined by
$$T_\psi(x)\in \argmax_{y\in Y} G(x,y,\psi(y))$$
transports $\mu$ onto $\nu$, i.e. $\nu$ is the image of the measure $\mu$ under $T_\psi$, denoted
$$T_{\psi\#}\mu = \nu.$$ Note that one needs to impose some conditions
on $\mu$ and $G$ ensuring that the map $T_\psi$ is well-defined
$\mu$-almost everywhere. One can describe the meaning of this equation
using an economic metaphor. We consider $X$ as a set of customers, $Y$
as a set of products and $G(x,y,\psi(y))$ corresponds to the utility
of the product $y$ for the customer $x$ given a price $\psi(y)$. The
probability measure $\mu$ and $\nu$ describe the distribution of
customers and products. The map $T_\psi$ can be described as the
``best response'' of customers given a price menu $\psi: Y\to\Rsp$:
each customer $x\in X$ tries to maximize its own utility $G(x,y,\psi(y))$ over
all products $y\in Y$: the maximizer, if it exists and is unique, is denoted $T_\psi(x)$. Then, $\psi$ is a solution to the generated jacobian
equation if the best response map $T_\psi$ pushes the distribution of
customers to the distribution of available products $\nu$.

In this article, we are interested in algorithms for solving the
semi-discrete case, where the source measure $\mu$ is absolutely
continuous with respect to the Lebesgue measure on $X\subseteq\Rsp^d$
and the target measure $\nu$ is finitely supported. Such
discretization can be traced back to Minkowski, but have been used
more recently to solve Monge-Ampère equations
\cite{oliker1989numerical}, problems from non-imaging optics
\cite{caffarelli1999problem}, more general optimal transport problems
\cite{kitagawa2014iterative}, but also generated Jacobian
equations~\cite{abedin2017iterative}. In all the cited papers, the
methods are coordinate-wise algorithms with minimal increment and are
similar to the algorithm introduced by
Oliker-Prussner~\cite{oliker1989numerical}. The number of iterations
of these algorithms scales more than cubicly ($N^3$, where $N$ is the
size of the support of $\nu$), making them limited to fairly small
discretizations.  More recently, Newton methods have been introduced
to solve semi-discrete optimal transport
problems~\cite{kitagawa2016newton,merigot2018algorithm}. In this
paper, we show that newtonian techniques can also be applied to
Generated Jacobian equations under mild conditions on the generating
function $G$.

\subsection*{Semi-discrete optimal transport.} 
The semi-discrete setting refers to the case where one is given an absolutely continuous probability measure $\mu$ (with respect to the Lebesgues measure) supported on a domain $X$ of $\R^d$ and a discrete probability measure $\nu=\sum_y \nu_y \delta_y$ supported on a finite set $Y$. Given a cost function $c : X \times Y \to \R$, the optimal transport problem amounts to finding a function $T : X \to Y$ that minimizes the total cost $\int_X c(x, T(x)) d\mu(x)$ under the condition $\mu(T^{-1}(y)) = \nu_y$ for any $y \in Y$. This problem can be recast,  using Kantorovitch duality under some mild conditions on the cost $c$, into finding a dual potential $\psi : Y \to \R$ that satisfies
\begin{equation}\label{eq:otdual}
\forall y \in Y\quad \mu(\Lag_y(\psi)) = \nu_y\ \tag{$\mathrm{MA}$}
\end{equation}
where $\Lag_y(\psi)$ are the Laguerre cells defined by
\[
\Lag_y(\psi) = \left\{ x \in X \mid \forall z \in Y, c(x,y) + \psi(y) \leq c(x,z) + \psi(z) \right\}.
\]
The application $T_\psi$ defined for $x \in X$ by $T_\psi(x) = y$ if $
x \in \Lag_y(\psi)$ is then an optimal transport map between $\mu$ and $\nu$ for the cost $c$, and satisfies in particular $T_{\psi\#} \mu = \nu$.  Equation~\eqref{eq:otdual} can
be regarded as a discrete version of the Monge-Amp\`ere type equation
arising in optimal transport. We refer for instance to
\cite[\S2]{berman2018convergence} for more details in the case $c(x,y) = -\sca{x}{y}$.

\subsection*{Generated Jacobian equation.}
The Generated Jacobian equation in the semi-discrete setting has a very similar form. The problem also amounts to finding a function $\psi:Y\to \R$ that satisfies Equation~\eqref{eq:otdual}, but the Laguerre cells have a more general form and read
\[ \Lag_y(\psi) = \left\{ x \in X \mid \forall z \in Y, G(x, y, \psi(y)) \geq G(x,z, \psi(z)) \right\} \]
where $G$ is called a \emph{generating function}. When $G$ is linear in the last variable, i.e. when $G(x,y,v) = -c(x,y) - v$, one obviously recovers the Laguerre cells from optimal transport. 

Note that the lack of linearity in the generating function $G$ adds
several theoretical and practical difficulties. To see this, consider  the mass function
$$\begin{aligned}
  H : \R^Y \to \R^Y, ~~
  \psi \mapsto (\mu(\Lag_y(\psi)))_{y \in Y}.
  \end{aligned}$$
In the optimal transport case, the function $H$ is invariant under the
addition of a constant (i.e. $H(\psi+c) = H(\psi)$ for any $c\in\Rsp$), which entails
under mild assumptions that the kernel of $\D H(\psi)$ has rank one
and coincides with the vector space of  constant
functions on $Y$~\cite{kitagawa2016newton}. Furthermore, as a consequence of
Kantorovitch duality, the function $H$ is the gradient of a
functional, called \emph{Kantorovitch functional} in
\cite{kitagawa2016newton}. This implies that the differential $\D
H(\psi)$ is symmetric. In the case of generated Jacobian equations,
these two properties do not hold anymore: the differential $\D
H(\psi)$ is not necessarily symmetric and its kernel is not reduced to
the set of constant functions in general. \\

In this article, we generalize the damped Newton algorithm proposed in~\cite{kitagawa2016newton} to solve generated Jacobian equations. Note that unlike~\cite{kitagawa2016newton} we do not require any Ma-Trudinger-Wang type condition to prove the convergence of our algorithm. In Section \ref{sec:GJE} we recall the notion of generating function and its properties, and introduce the generated Jacobian equation in the semi-discrete setting.  
Section \ref{sec:num} is entirely dedicated to the numerical resolution of the generated Jacobian equation. 
In Section~\ref{sec:application}, we apply our algorithm to numerically solve the Near Field Parallel Reflector problem. Note that F. Abedin and C. Gutierrez also consider this problem~\cite{abedin2017iterative}, but their algorithm requires a strong condition, called \emph{Visibility Condition}, that implies the \emph{Twist condition} (defined hereafter) of the generating function $G$. 
We show that under a much weaker assumption, this twist condition holds for a subset of dual potential $\psi:Y\to \R$ on which we can apply our algorithm. It is very likely that our assumption could also be adapted to~\cite{abedin2017iterative}.

\section{Semi-discrete generated Jacobian equation}\label{sec:GJE}
In this section, we recall the notions introduced by N. Trudinger in order to define the generated Jacobian equation~\cite{trudinger2012local} in the semi-discrete setting. 
Let $\Omega$ be an open bounded domain of $\R^d$, let $X$ be a compact subset of $\Omega$ and let $Y$ be a finite set of $\R^d$. 
Let $\mu$ be a measure on $\Omega$, which is absolutely continuous with respect to the Lebesgue measure, with non-negative density $\rho$ supported on $X$ (i.e. $\supp(\rho) \subset X$), and let $ \nu = \sum_{y\in Y} \nu_y \delta_{y}$ be a measure on the finite set $Y$ such that all $\nu_y$ are positive ($\nu_y > 0$).  These two measures must satisfy the mass balance condition $\mu(X) = \nu(Y)$ and it is not restrictive to view them as probability measures:
\begin{equation*}
	\int_{X} \rho (x) dx = \sum_{y\in Y} \nu_y= 1
\end{equation*}

\noindent \textbf{Notations.} We denote by $\mathcal{H}^k$ the $k$-dimensional Hausdorff measure in $\R^d$. In particular $\mathcal{H}^d$ is the Lebesgue measure in $\R^d$. The set of functions from $Y$ to $\R$ is denoted by $\R^Y$.  We denote by $\sca{\cdot}{\cdot}$  the Euclidean scalar product, by $\| \cdot \|$ the Euclidean norm, by $\B(x,r)$ the Euclidean  ball of center $x$ and radius $r$, by $\chi_A : \R^d \to \{0,1\}$ the indicator function of a set $A$. The image and kernel of a matrix $M$ are respectively denoted by  $\im(M)$ and $\ker(M)$. We denote by $\Span(u)$ the linear space spanned by a vector $u$, by $\nabla_xG$ the gradient of a function $G$ with respect to $x$ and by $\partial_vG$ its scalar derivative with respect to $v$. Finally, for $N\in\Nsp$, we denote $\OneN = \{1,\hdots,N\}$. 

\subsection{Generating function}\label{sec:genfunc}
We recall below the notion of generating function and $G$-convexity in the semi-discrete setting~\cite{trudinger2012local,abedin2017iterative}. 
\begin{definition}[Generating function]
	\label{defG}
	Let $a,b \in \R \cup \{ -\infty, + \infty\}$ with $a < b$ and
        $I = ]a,b[$. A function $G : \Omega \times Y \times I \to \R$
            is called a generating function. We assume that it
            satisfies the following properties:
	\begin{itemize}
		\item \emph{Regularity condition:} 
		$(x,y,v) \mapsto G(x,y,v)$ is continuously differentiable in $x$ and $v$, and
		\[  \forall \alpha < \beta \in I, \sup_{(x,y,v) \in \Omega \times Y \times [\alpha,\beta]} |\nabla_xG(x,y,v)| < + \infty \label{Regularity} \tag{Reg} \]
		\item \emph{Monotonicity condition:}  \[ \forall (x,y,v) \in \Omega \times Y \times I : \partial_vG(x,y,v) < 0 \label{Monotonicity} \tag{Mono} \]
		\item \emph{Twist condition:} \[ \forall x \in \Omega, (y,v) \mapsto (G(x,y,v), \nabla_xG(x,y,v)) \text{ is injective on } Y \times I \label{Twist} \tag{Twist}\]
		\item \emph{Uniform Convergence condition:} \[\forall y \in Y, \lim_{v \to a} \inf_{x \in \Omega} G(x,y,v) = + \infty \label{Uniform Convergence} \tag{UC} \]
	\end{itemize}
\end{definition}
\begin{remark}[Range of $G$]
	\label{interval}
	Through the whole paper we can and will consider that $I = \R$.
	Indeed suppose that  $G : \Omega \times Y \times I \to \R$ satisfies the assumptions of the above definition. Considering  a strictly increasing $\Class^1$ diffeomorphism $\zeta : \R \to I$ and setting  $\widetilde{G}(x,y,v) = G(x,y, \zeta(v))$, we get a generating function $\tilde{G}: \Omega\times Y \times \Rsp \to \Rsp$, which also satisfies the conditions above. Moreover, up to reparametrization, the generated Jacobian equations associated to $G$ and $\tilde{G}$ are equivalent.
\end{remark}
\begin{remark} Note that F. Abedin and C. Gutierrez \cite{abedin2017iterative} impose a slightly more restrictive inequality in condition~\eqref{Regularity}: their supremum is taken over $\Omega \times Y \times ]- \infty, \alpha]$ for any $\alpha$ instead of $\Omega \times Y \times [\alpha,\beta]$. We changed here this condition in order to handle the Near Field point reflector problem in the last section.
\end{remark}


\begin{definition}[G-convexity]
  \label{G-convexity}
  Let $\phi: \Omega\to\R$ be a function. If $\phi \geq G(\cdot,y_0, \lambda_0)$ for all $x\in\Omega$ with equality at $x=x_0$, we say that the function $G(\cdot,y_0,\lambda_0)$ \emph{supports} $\phi$ at $x_0$. 
	A function $\phi : \Omega \to \R$ is said to be G-convex if it is supported at every point, i.e. for all $x_0 \in \Omega$,
	\begin{equation}
		\label{G-conv}
		\exists (y_0, \lambda_0) \in Y \times \R \text{ s.t. } 
		\begin{cases}
			\forall x \in \Omega, \phi(x) \geq G(x, y_0, \lambda_0) \\
			\phi(x_0) = G(x_0, y_0,\lambda_0)
		\end{cases}
	\end{equation}
\end{definition}

\begin{remark}[Relation with convexity]
	The notion of G-convexity generalizes in a certain sense the
        notion of convexity. Intuitively, it amounts to replacing the
        supporting hyperplanes by functions of the form $G(\cdot, y,
        \lambda)$. If $G(\cdot, y, \lambda)$ is convex for any $y\in
        Y$ and any $\lambda\in \R$, then a $G$-convex function is
        always convex. Moreover, if the generating function $G$ is
        affine (i.e $G(x,y,\lambda) = \langle x,y \rangle + \lambda$)
        and if $Y=\R^d$, then the notions of G-convexity and convexity
        are equivalent.
\end{remark}

\begin{definition}[G-subdifferential]
	Let $\phi$ be a G-convex function and let $x_0 \in \Omega$. The G-subdifferential $\partial_G \phi$ of $\phi$ at $x_0$ is defined by
	\begin{equation}
		\label{Gsub}
		\partial_G \phi (x_0) = \left\{ y \in Y \mid \exists \lambda_0 \in \R \text{s.t. $G(\cdot , y , \lambda_0)$ supports $\phi$ at $x_0$}  \right\}
	\end{equation}	
\end{definition}

The following lemma (Lemma~2.1 in \cite{abedin2017iterative}) shows
that the $\partial_G\phi$ is single-valued almost everywhere, and
induces a measurable 

\begin{lemma} \cite[Lemma 2.1]{abedin2017iterative}
	\label{uniqueness-a-e}
	Let $\phi$ be G-convex with $G$ satisfying \eqref{Regularity}, \eqref{Monotonicity} and \eqref{Twist}. Then, there exists a measurable map $\Tphi:\Omega\to Y$ s.t.
        $$ \hbox{ for a.e. } x \in \Omega,\quad \partial_G\phi(x) = \{ \Tphi(x) \}.$$
	\end{lemma}

We can define the notion of generated Jacobian equation. 
\begin{definition}[Brenier solution to the GJE]\label{def:brenier}
	A function $\phi : X \to \R$ is a Brenier solution to the  generated Jacobian equation between a probability density $\mu$ on $\Omega$ and a probability measure $\nu = \sum_{y\in Y} \nu_y\delta_{y}$ on $Y$ if it satisfies
\begin{equation}
\label{GenJac}
\tag{GenJac}
\begin{cases}
\phi \text{ is G-convex}\\
\forall y \in Y ,  \mu(\Tphi^{-1}(\{y\})) = \nu_y
\end{cases}
\end{equation}
\end{definition}

\subsection{G-transform}\label{sec:genjac}
The goal in this section is to write a dual formulation of the
generated Jacobian equation, using the notion of $G$-transform
introduced by Trudinger \cite{trudinger2012local}.
\begin{definition} 
	 The $G$-transform $\psi^G : \Omega \to \R$ of $\psi: Y\to\Rsp$ is defined by
	\begin{equation}
	\label{psiG}
	\forall x \in \Omega,~~ \psi^G(x) = \max_{y \in Y} G(x, y, \psi(y)).
	\end{equation}
\end{definition}

\begin{proposition}
\label{prop-dual}
	Assume $G$ satisfies \eqref{Regularity}, \eqref{Monotonicity} and \eqref{Twist} and let $\phi : \Omega \to \R$ be a $G$-convex function. Then there exists $\psi : \Tphi(\Omega) \to \R$ s.t. \[ \forall x \in \Omega,~~ \phi(x) = \max_{y \in \Tphi(\Omega)} G(x, y, \psi(y)) \]
\end{proposition}
\begin{proof}
	Let $y \in \Tphi(\Omega)$, then for any $x_0 \in \Tphi^{-1}(y)$ there exists $\lambda_0 \in \R$ such that $\phi(x_0) = G(x_0, y, \lambda_0)$. Since $\phi$ is $G$-convex we also have for any $x \in \Omega$ that $\phi(x) \geq G(x,y,\lambda_0)$. Specifically for $x_1 \in \Tphi^{-1}(y)$, we get $\phi(x_1) = G(x_1, y, \lambda_1) \geq G(x_1, y, \lambda_0)$ and since $\partial_vG(x,y,v) < 0$ then $\lambda_1 \leq \lambda_0$. By symmetry we have $\lambda_1 = \lambda_0$. We can deduce that there exists a unique $\psi(y) \in \R$ such that for any $x \in \Tphi^{-1}(y), \phi(x) = G(x,y,\psi(y))$. 
	This defines a map $\psi : \Tphi(\Omega) \to \R$ satisfying
	\[ \forall x \in \Omega,~ 
	\begin{cases}
	\forall y \in \Tphi(\Omega) , \phi(x) \geq G(x,y,\psi(y)) \\
	\exists y \in \Tphi(\Omega) , \phi(x) = G(x,y,\psi(y))
	\end{cases} 
	\]
	As a conclusion we have $\displaystyle \phi(x) = \max_{y \in \Tphi(\Omega)} G(x, y, \psi(y))$.
\end{proof}
\begin{corollary}
	\label{corollary-dual}
	Let $\phi$ be a $G$-convex function such that $\Tphi(\Omega) = Y$, then there exists $\psi : Y \to \R$ such that $\phi = \psi^G$.
\end{corollary}

\begin{remark}[$G$-convex functions are not always $G$-transforms]
	\label{rem:Gtransform}
	Without any additional assumptions on the generating function,
        we cannot guarantee that any $G$-convex function $\phi$ on $X$
        is the $G$-transform of a function $\psi$ on $Y$. Define for
        instance
        $$\Omega = (1,2), \quad Y = \{0, 1\}, \quad G(x,y, v) = \begin{cases} xe^{-v} & \hbox{ if } y = 0\\
          -xv & \hbox{ if } y=1\end{cases}.$$
          and consider the function $\phi$ on $\Omega$ defined by $\phi(x) = G(x,1,1) = -x$, which is  $G$-convex by definition. Yet for any $v \in \R$ and any $x \in \Omega$,
          $$\max(G(x,0,v), G(x,1,1)) = \max(xe^{-v}, -x) = x e^{-v},$$
          thus implying that that there does not exist any $\psi: Y
          \to \R$ such that $\phi$ is the $G$-transform of $\psi$.

\end{remark}

Suppose that $\phi$ is a solution of \eqref{GenJac} and that for all
$y\in Y$, the mass $\nu_y$ is positive. Then, for any $y \in Y$ one
has $\mu(\Tphi^{-1}(y)) = \nu_y > 0$, which guarantees that
$\Tphi(\Omega) = Y$. Therefore by Corollary \ref{corollary-dual} there
exists a function $\psi$ on $Y$ such that $\phi = \psi^G$. This means
that we can reparametrize the problem \eqref{GenJac} by assuming that
the solution $\phi$ is the $G$-transform of some function $\psi$. The sets
$S_{\psi^G}^{-1}(\{y\})$, which appear in \eqref{GenJac} will be
called generalized Laguerre cells.

\begin{definition}[Generalized Laguerre cells] The \emph{generalized Laguerre cells}  associated to a function $\psi : Y \to \R$ are defined for every $y\in Y$ by
  \begin{equation}\label{Lagi}
\begin{aligned}
\Lag_y(\psi) &:= S_{\psi^G}^{-1}(\{y\})\\ 
&= \left\{ x \in \Omega \mid \forall z \in Y , G(x, y, \psi(y)) \geq G(x, z, \psi(z))\right\}.
\end{aligned}
\end{equation}
\end{definition}

Note that by Lemma~\ref{uniqueness-a-e}, the intersection of two generalized
Laguerre cells has zero Lebesgue measure, ensuring that the sets
$\Lag_y(\psi)$ form a partition of $\Omega$ up to a $\mu$-negligible set. 

\begin{definition}[Alexandrov solution to GJE]
  A function $\psi : Y \to \R$ is an \emph{Alexandrov solution} to the generated Jacobian equation between  generated Jacobian equation between a probability density $\mu$ on $\Omega$ and a probability measure $\nu = \sum_{y\in Y} \nu_y\delta_{y}$ on $Y$ if $\psi^G$ is a Brenier solution (Definition~\ref{def:brenier}) to the same GJE, or equalently if
  $$ 	\forall y\in Y,~~ H_y(\psi) = \nu_y, \quad \hbox{ where } H_y(\psi) = \mu(\Lag_y(\psi)).$$
  Setting $H(\psi) = (H_y(\psi))_{y \in Y}$ and considering $\nu$ as a function over $Y$, we can even rewrite this equation as
  \begin{equation}
    \label{GenJac'}
    H(\psi) = \nu.
	\tag{GenJacD}
\end{equation}
\end{definition}

\section{Resolution of the generated Jacobian equation}\label{sec:num}
The goal of this section is to introduce and study a Newton algorithm
to solve the semi discrete generated Jacobian
equation~\eqref{GenJac'}. Before doing so, we study the regularity of
the mass function $H:\R^Y\to \R^Y$ in Section~\ref{section:regularity}
and establish a non-degeneracy property of its differential $\D H$ in
Section~\ref{section:differential}, under a connectedness assumption
on the support of the source measure. We present the algorithm and
prove its convergence in Section~\ref{sec:algo}.

For simplicity, we will number the points in $Y$, i.e. we assume that $$Y = \{ y_1,\hdots,y_N\},$$ where the points $y_i$ are distinct. This allows us to  identify the set of
functions $\R^Y$ with $\R^N$, by setting $\psi_i = \psi(y_i)$. We also
denote $(e_i)_{1 \leq i \leq N}$ the canonical basis of $\R^N$. Finally, we introduce a shortened notation for Laguerre cells and intersections thereof
$$\Lag_i(\psi) = \Lag_{y_i}(\psi),\quad \Lag_{ij}(\psi) = \Lag_i(\psi)
\cap \Lag_j(\psi).$$ Throughout this section, we assume that the
generating function $G$ satisfies all the conditions of
Definition~\ref{defG}.

\subsection{$\mathcal{C}^1$-regularity of $H$}\label{section:regularity}
The differentiability of $H$ is established under a (mild) genericity hypothesis on the cost function, ensuring in particular that the intersection between three distinct Laguerre cells is negligible with respect to the $(d-1)$-dimensional Hausdorff measure, denoted $\mathcal{H}^{d-1}$. To write this hypothesis,  we denote for three distinct indices  $i,j,k$ in $\OneN$,
$$\Gamma_{ij}(\psi) = \{x \in \Omega \mid G(x, y_i, \psi_i) = G(x, y_j, \psi_j) \},\quad \Gamma_{ijk}(\psi) = \Gamma_{ij}(\psi) \cap \Gamma_{ik}(\psi).$$

\begin{definition}[Genericity of the generating function.]
	\label{Ygeneric}
        The  generating function $G$ is generic with respect to $\Omega$ and $Y$ if for any distinct indices $i,j,k$ in $\OneN$  and any $\psi \in \R^N$ we have
        \begin{equation} \label{GenOmega}\tag{$\mathrm{Gen}_\Omega^Y$}
	  \mathcal{H}^{d-1}(\Gamma_{ijk}(\psi)) = 0.
        \end{equation}
The generating function $G$ is generic with respect to the boundary
$\partial X$ and $Y$ if for any distinct indices $i,j$ in $\OneN$ and any $\psi \in \R^N$ we
have
\begin{equation}
  \label{GenX}\tag{$\mathrm{Gen}_{\partial X}^Y$}
  \mathcal{H}^{d-1}(\Gamma_{ij}(\psi) \cap \partial X) = 0.
\end{equation}
\end{definition}

\begin{proposition}
  \label{Hi_diff}
  Assume that
  \begin{itemize}
    \item $G \in \Class^2(\Omega\times Y \times \R)$ satisfies
  \eqref{Regularity}, \eqref{Monotonicity}, \eqref{Twist},
  \eqref{GenOmega}, \eqref{GenX},
  \item  $X \subseteq \Omega$ is compact
    and that $\rho$ is a continuous probability density on $X$.
  \end{itemize}
  Then the
  mass function $H : \R^N \to \R^N $ defined by $H(\psi) =
  (\mu(\Lag_i(\psi)))_{1\leq i \leq N}$ has class
  $\mathcal{C}^1$. We have for $\psi \in \R^N$ and $i \in \llbracket
  1,N \rrbracket$
	\begin{equation}
	\label{partialH}
	\begin{cases}
	\displaystyle \frac{\partial H_j}{\partial \psi_i}(\psi) = \int_{\Lag_{ij}(\psi)} \rho (x)\frac{|\partial_vG(x,y_i,\psi_i )|}{\Vert\nabla_xG(x,y_j,\psi_j) - \nabla_xG(x,y_i,\psi_i )\Vert}d\mathcal{H}^{d-1}(x) \geq 0 \text{ for } j \neq i \\
	\displaystyle \frac{\partial H_i}{\partial \psi_i}(\psi) = - \sum_{j\neq i} \frac{\partial H_j}{\partial \psi_i}(\psi) 
	\end{cases}
	\end{equation}
\end{proposition}

\begin{proof}
	Let $\psi \in \R^N$ and $i,j \in \OneN$ be fixed indices such that $i \neq j$. 
 	We want to compute $\partial H_j/\partial \psi_{i}(\psi)$. For this purpose, we introduce
        $\psi^t = \psi + te_{i}$ for $t\geq 0$.
	From \eqref{Monotonicity}, we obviously have $\Lag_j(\psi) \subseteq \Lag_j(\psi^t)$. 
	Therefore 
	$$H_j(\psi^t) - H_j(\psi) = \mu(\Lag_j(\psi^t)) - \mu(\Lag_j(\psi) ) = \mu(\Lag_j(\psi^t) \setminus \Lag_j(\psi) )$$
	We introduce the set $L$ obtained by removing one inequality in the definition of
        the generalized Laguerre cell $\Lag_j(\psi)$:
        $$L = \{x \in \Omega \mid \forall k \neq i , G(x,y_j, \psi_j)
        \geq G(x,y_k,\psi_k)\}.$$ We have in particular $\Lag_j(\psi)
        \subseteq L$ and more precisely
        $$ \Lag_j(\psi^t) \setminus \Lag_j(\psi) = \bigsqcup_{0 < s \leq t} L \cap \Gamma_{ij}(\psi^t).$$
        We will use this formula to get another expression of $H_j(\psi^t) - H_j(\psi)$.
        \medskip

        \noindent \textbf{Step 1. Construction of $u_{ij}$ such that $\Gamma_{ij}(\psi^t) = u_{ij}^{-1}(\{t\})$.}
        To construct such a function $u_{ij}:\Omega\to\R$, we first consider  the function $f_{ij} : \Omega \times \R \to
            \R$ defined by
	\begin{equation*}
	f_{ij}(x,t) = G(x,y_j,\psi_j) - G(x,y_i,\psi_i +t)
	\end{equation*}
        This function $f_{ij}$ is of class $\mathcal{C}^1$ on $\Omega \times \R$ by hypothesis on $G$ and we have 
        $$ \forall (x,t) \in \Omega \times \R,  \frac{\partial f_{ij}}{\partial t} (x,t) = -\partial_vG(x,y_i,\psi_i + t) > 0.$$
        This implies that a fixed $x \in\Omega$, the function $f_{ij}(x,\cdot)$ is strictly increasing, so that equation $f_{ij}(x,t) = 0$ has at most one solution. Denoting
        $$\mathcal{V}_{ij} = \{ x\in \Omega\mid \exists t\in\Rsp, 
        f_{ij}(x,t) = 0\} = \bigcup_{t \in \R} \Gamma_{ij}(\psi^t),$$
        one can therefore define a function $u_{ij}: \mathcal{V}_{ij}
        \to \Rsp$  which satisfies 
$$ \forall x\in \mathcal{V}_{ij},~~ f_{ij}(x, t) = 0 \Longleftrightarrow u_{ij}(x) = t.$$
	By the implicit function theorem, the set $\mathcal{V}_{ij}$ is open and the function $u_{ij}$ is  $\mathcal{C}^1$ on $\mathcal{V}_{ij}$.
        In order to apply the co-area formula, we need to compute the gradient of $u_{ij}$.
	For any point $x$ in $\mathcal{V}_{ij}$, we have by definition
        $$f_{ij}(x, u_{ij}(x)) = G(x,y_j,\psi_j) - G(x,y_i,\psi_i + u_{ij}(x)) = 0.$$
        Differentiating this expression with respect to $x$, we obtain
	$$\nabla u_{ij}(x) = \frac{\nabla_xG(x,y_j,\psi_j) -
          \nabla_xG(x,y_i,\psi_i +
          u_{ij}(x))}{\partial_vG(x,y_i,\psi_i +u_{ij}(x))}$$ which is
        well defined since $\partial_vG(x,y_i,v) < 0$ on $\Omega
        \times Y \times \R$ by the \eqref{Monotonicity} hypothesis.
        The \eqref{Twist} condition guarantees that for all $x \in
        \mathcal{V}_{ij}$, the map $(y,v) \mapsto
        (G(x,y,v),\nabla_xG(x,y,v))$ is injective. By definition of
        $u_{ij}$ we have $f_{ij}(x, u_{ij}(x))$, so that
        $$G(x,y_j,\psi_j) = G(x,y_i,\psi_i + u_{ij}(x)).$$
        The \eqref{Twist} condition then entails $$\nabla_xG(x,y_j,\psi_j) \neq \nabla_xG(x,y_i,\psi_i + u_{ij}(x)),$$
        implying that the gradient $\nabla u_{ij}(x)$ does not vanish. \medskip

        \noindent \textbf{Step 2. Computation of the partial derivatives.}
        We can write the difference between Laguerre cells using the function $u_{ij}$:
	\begin{align*}
		\Lag_j(\psi^t) \setminus \Lag_j(\psi) &= \bigcup_{0 < s \leq t} \Lag_{ij}(\psi^s) \\
		&= \left\{ x \in \Omega , \exists s \in ]0,t], f_{ij}(x,s) = 0 \right\} \cap L \\
		&= \left\{ x \in \Omega , \exists s \in ]0,t], u_{ij}(x) = s \right\} \cap L \\
		&= u_{ij}^{-1} (]0,t]) \cap L,
	\end{align*}
	giving directly
	$$H_j(\psi^t) - H_j(\psi) = \mu(L \cap u_{ij}^{-1}(]0,t])) = \int_{L \cap u_{ij}^{-1}(]0,t])} \rho (x) dx.$$
	Then the co-area formula gives us
	$$\frac{H_j(\psi^t) - H_j(\psi)}{t} = \frac{1}{t}  \int_{L \cap u_{ij}^{-1}(]0,t])} \rho (x) dx= \frac{1}{t} \int_{0}^t H_{ij}(\psi^s)ds, $$
	where we introduced
	\begin{equation}
		\label{hij}
		\displaystyle H_{ij}(\psi) = \int_{\Lag_{ij}(\psi)} \frac{\rho (x)}{\Vert\nabla u_{ij}(x)\Vert}d\mathcal{H}^{d-1}(x).
	\end{equation}
        Note that thanks to the computations above, we already know
        that the gradient $\nabla u_{ij}(x)$ does not
        vanish. Moreover, for any $x$ in $\Lag_{ij}(\psi)\subseteq \Gamma_{ij}(\psi)$, one has 
        $u_{ij}(x) = 0$. Thus,
        $$\nabla u_{ij}(x) =
        (\nabla_xG(x,y_j,\psi_j) - \nabla_xG(x,y_i,\psi_i )) /
        (\partial_vG(x,y_i,\psi_i )). $$
        We can therefore rewrite 
\begin{equation}
		\displaystyle H_{ij}(\psi) = \int_{\Lag_{ij}(\psi)} \frac{\rho (x) |\partial_vG(x,y_i,\psi_i)|}{\vert\nabla_xG(x,y_j,\psi_j) - \nabla_xG(x,y_i,\psi_i)\vert} d\mathcal{H}^{d-1}(x).
\end{equation}
As shown in Proposition~\ref{hcont} below, $H_{ij}$ is continuous on $\R^N$. 
	We deduce that
	\begin{equation}
	\label{dH}
	\frac{\partial H_j}{\partial \psi_i}(\psi) = \lim_{t \to 0,t>0} \frac{H_j(\psi^t) - H_j(\psi)}{t} = H_{ij}(\psi) \geq 0.
	\end{equation}
	The case $t < 0$ can be treated similarly by replacing
        $\Lag_j(\psi) \subseteq \Lag_j(\psi^t)$ with $\Lag_j(\psi^t)
        \subseteq \Lag_j(\psi)$. We thus get the desired expression
        for the partial derivative $\partial H_j/\partial \psi_i$ for
        $i\neq j$.

	To compute the partial derivative for $j=i$, we use the mass conservation property
        $\sum_{1\leq i\leq N} H_i(\psi) = 1$ to  deduce that \begin{equation*}
          \frac{\partial H_i}{\partial \psi_i}(\psi) =- \sum_{j\neq i} \frac{\partial H_j}{\partial \psi_i}(\psi). \qedhere\end{equation*}
\end{proof}

It remains to show that the functions $H_{ij}$ used in the proof of Proposition \ref{Hi_diff} are continuous.

\begin{proposition}
	\label{hcont}
	Under the assumptions of Proposition~\ref{Hi_diff}, 
	for every $i,j \in \OneN$, the function $H_{ij}$ defined in \eqref{hij} is continuous on $\R^N$. 
\end{proposition}
\begin{proof}
	We introduce the function $g : \Omega \times  \R^N  \to \R$ defined by
	 \[ g(x,\psi) = \bar{\rho}(x) \frac{|\partial_vG(x,y_i,\psi_i )|}{\Vert\nabla_xG(x,y_j,\psi_j) - \nabla_xG(x,y_i,\psi_i )\Vert} \]
	 where $\bar{\rho}$ is a continuous extension of the probability density $\rho_{|X}$ on $\Omega$.
	 For a given $\psi \in \R^N$, the \eqref{Twist} hypothesis guarantees that for any $x \in \Gamma_{ij}(\psi)$, $\nabla_xG(x,y_j,\psi_j) \neq \nabla_xG(x,y_i,\psi_i)$.
	 This implies that $g$ is continuous on a neighborhood of the set $\{ (x, \psi) \in \Omega \times \R^N | x \in \Gamma_{ij}(\psi) \}$. 
	 We introduced in Proposition~\ref{Hi_diff} the function
	 \[ H_{ij}(\psi) = \int_{\Lag_{ij}(\psi) \cap X } g(x,\psi) d\mathcal{H}^{d-1}(x). \]
	 Let $\psi^\infty \in \R^N$ and $\psi^n$ a sequence converging towards $\psi^\infty$.The main difficulty for proving that $H_{ij}(\psi^n)$ converges to $H_{ij}(\psi^\infty)$ as $n\to+\infty$ is that the integrals in the definition of $H_{ij}(\psi^n)$ and $H_{ij}(\psi^\infty)$ are over different hypersurfaces, namely $\Gamma_{ij}(\psi^n)$ and $\Gamma_{ij}(\psi^\infty)$. Our first step will therefore be to construct a diffeomorphism between (subsets) of these hypersurfaces. We introduce  $f : \R \times \R \times \Omega  \to \R$ the function defined by
	 \[ f(a,b,x) = G(x,y_j,\psi^\infty_j + a) - G(x,y_i, \psi^\infty_i + b)\]
	 We put $a_n = \psi_j^n - \psi_j^\infty$ and $b_n = \psi_i^n - \psi_i^\infty$, so that  $a_n \to 0$ and $b_n \to 0$ as $n$ tends to $+\infty$. We also have
         $$\Gamma_{ij}(\psi^\infty) = (f(0,0, \cdot))^{-1}(0), \quad \Gamma_{ij}(\psi^n) = (f(a_n,b_n, \cdot))^{-1}(0).$$

         
	 \noindent \textbf{Step 1: Construction of a map $F_n$ between $\Gamma_{ij}(\psi^\infty)$ and $\Gamma_{ij}(\psi^n)$.}\\
	 This map is constructed using the composition of the flows associated to two vector fields $X_a$ and $X_b$.
	 Let $\widetilde{\Omega} \subset \Omega$ an open domain containing $X$. 
	 The \eqref{Twist} hypothesis guarantees that there exists a neighborhood $\widetilde{V}$ of the set $\{ (a,b,x) \in \R^2 \times \widetilde{\Omega} | f(a,b,x) = 0 \}$ such that we have for any $v \in \widetilde{V}$, $\nabla_x f(v) \neq 0$.
	 We can then define two vector fields $X_a, X_b$ on $\widetilde{V}$ by
	 \begin{align*}
	 X_a (a,b,x) &= \left( 1, 0, -\partial_a f (a, b, x) \frac{\nabla_x f(a, b, x)}{\| \nabla_x f(a, b, x) \|^2} \right) \\
	 X_b (a,b,x) &= \left( 0, 1, -\partial_b f (a, b, x) \frac{\nabla_x f(a, b, x)}{\| \nabla_x f(a, b, x) \|^2} \right)
	 \end{align*}
	 Since $f$ is of class $\Class^2$, $X_a$ and $X_b$ are both of class $\Class^1$ on $\widetilde{V}$. We then consider $\Phi_a$ and $\Phi_b$ the flows associated respectively to $X_a$ and $X_b$ defined for  $(t,v) \in [-\varepsilon, \varepsilon]^2 \times \widetilde{V}$ by
	 \[
	 \begin{cases}
	 	\Phi_a(0, v) = v \\
	 	\partial_t \Phi_a(t, v) = X_a(\Phi(t,v))
	 \end{cases}
	 \]
	 and
	 \[
	 \begin{cases}
	 	\Phi_b(0, v) = v \\
	 	\partial_t \Phi_b(t, v) = X_b(\Phi(t,v))
	 \end{cases}
	 \]
	 The vector fields $X_a$ and $X_b$ are continuously differentiable on $\widetilde{V}$ which implies that both $\Phi_a(t, \cdot)$ and $\Phi_b(t, \cdot)$ converge pointwise in the $\Class^1$ sense toward the identity as $t \to 0$.
	 Let $(t,v) \in [-\varepsilon, \varepsilon] \times \widetilde{V}$. Denoting $\nabla f (v) = (\partial_a f, \partial_b f, \nabla_x f)(v)$, we then have
	 \begin{align*}
	 	f(\Phi_a(t,v)) &= f(\Phi_a(0,v)) + \int_0^t \frac{\partial}{\partial s} \big(s \mapsto f(\Phi_a(s,v)) \big) ds\\
	 	&= f(v) + \int_0^t \sca{\nabla f(\Phi_a(s,v))}{\partial_t \Phi_a(s,v)} ds\\
	 	&= f(v) + \int_0^t \sca{\nabla f(\Phi_a(s,v))}{X_a(\Phi_a(s,v))} ds \\
	 	&= f(v)
	 \end{align*}
	 Similarly one has $f(\Phi_b(t,v)) = f(v)$.
	 Let $\Pi : \widetilde{V} \to \Omega$ the projection of $\widetilde{V} \subseteq \Rsp^2\times\Omega$ on $\Omega$, and let $F_n : \Gamma_{ij}(\psi^\infty) \cap \widetilde{\Omega} \to \Omega$ be the function defined by
	 \[ F_n(x) = \Pi \big(\Phi_a(a_n,\Phi_b(b_n,(0,0,x))) \big).\]
	 For $x \in \Gamma_{ij}(\psi^\infty)$ and $v = (0,0,x) \in \widetilde{V}$,  we have
	 \[ \Phi_a(a_n,\Phi_b(b_n,v)) = (a_n, b_n, F_n(x))\]
	 and from the previous equality we deduce that
	 \[f(\Phi_a(a_n,\Phi_b(b_n,v))) = f(v) = 0 \]
	 This means that for $x \in \Gamma_{ij}(\psi^\infty)$, $F_n(x) \in \Gamma_{ij}(\psi^n)$.
	 Moreover $\Phi_a(a_n, \cdot)$ and $\Phi_b(b_n, \cdot)$ are both invertible of inverse $\Phi_a(-a_n, \cdot)$ and $\Phi_b(-b_n, \cdot)$. Thus $F_n$ is also invertible of inverse 
	 \[ F_n^{-1}(x) = \Pi \big( \Phi_b(-b_n, \Phi_a(-a_n, (a_n,b_n,x) ) ) \big)\]
	 \\
	 Since both $\Phi_a(a_n, \cdot)$ and $\Phi_b(b_n, \cdot)$ converge pointwise in the $\Class^1$ toward the identity as $n \to + \infty$, we have for $x \in \Gamma_{ij}(\psi^\infty) \cap \widetilde{\Omega}$
	 \[
	 \begin{cases}
	 	\displaystyle \lim_{n \to + \infty} F_n(x) = x, \\
	 	\displaystyle \lim_{n \to + \infty} JF_n(x) = 1,
	 \end{cases} 
	 \]
         where  $JF_n$ is the absolute value of the determinant of the Jacobian matrix of $F_n$.
         
	 \textbf{Step 2: Convergence of $H_{ij}(\psi^n)$ toward $H_{ij}(\psi^\infty)$.}\\
	 We let $L_\infty = \Lag_{ij}(\psi^\infty)$ and $L_n = F_n^{-1}(\Lag_{ij}(\psi^n) \cap \widetilde{\Omega})$. 
	 Denoting by $\chi_A$ the indicator function of a set $A$, we have
	 \begin{align*}
		 H_{ij}(\psi^\infty) &= \int_{\Gamma_{ij}(\psi^\infty)} g(x,\psi^\infty) \chi_X(x)\chi_{L_\infty}(x) d\mathcal{H}^{d-1}(x) \\
		 &= \int_{\Gamma_{ij}(\psi^\infty) \cap  \widetilde{\Omega}} g(x,\psi^\infty) \chi_X(x)\chi_{L_\infty}(x) d\mathcal{H}^{d-1}(x)
	 \end{align*}
	 because $\Gamma_{ij}(\psi^\infty) \cap X \subset \widetilde{\Omega}$.
	 We also have
	 \[ H_{ij}(\psi^n) = \int_{\Gamma_{ij}(\psi^n)} g(x,\psi^n) \chi_X(x)\chi_{\Lag_{ij}(\psi^n)}(x) d\mathcal{H}^{d-1}(x) \]
	 By a change of variable from $x$ to $F_n(x)$, the latter equality becomes
	 \[ H_{ij}(\psi^n) = \int_{\Gamma_{ij}(\psi^\infty) \cap  \widetilde{\Omega}} g(F_n(x),\psi^n) JF_n(x) \chi_X(F_n(x))\chi_{L_n}(x) d\mathcal{H}^{d-1}(x) \]
	 where $JF_n(x)$ denotes the determinant of the Jacobian
         matrix of $F_n$.  We already have the pointwise convergences
         $F_n(x) \to x$ and $JF_n(x) \to 1$ as $n \to \infty$.  If we
         can show that
	 \[ \lim_{n \to + \infty} \chi_X(F_n(x))\chi_{L_n}(x) =  \chi_X(x)\chi_{L_\infty}(x) \]
         for $\Haus^{d-1}$ almost every point $x$, 
	 then using Lebesgue's dominated convergence theorem, we will obtain that $H_{ij}(\psi^n) \to H_{ij}(\psi^\infty)$.

         We first show that $\lim_{n\to +\infty}\chi_{L_n}(x) \to \chi_{L_\infty}(x)$ $\Haus^{d-1}$-almost everywhere on $\Gamma_{ij}(\psi^\infty) \cap \widetilde{\Omega}$.
	 We first consider the superior limit: given  $x \in \Gamma_{ij}(\psi^\infty) \cap \widetilde{\Omega}$, we prove that $\limsup_{n \to \infty} \chi_{L_n}(x) \leq \chi_{L_\infty}(x)$. 
	The limsup is non-zero if and only if there exists a subsequence $(\sigma(n))_{n \in \N}$ such that $\forall n \in \N, x \in L_{\sigma(n)}$. 
	In this case we have $F_{\sigma(n)}(x) \in F_{\sigma(n)}(L_{\sigma(n)}) = \Lag_{ij}(\psi^{\sigma(n)}) \cap \widetilde{\Omega}$. 
	This means that for any $k \neq i,j$
	\[ G(F_{\sigma(n)}(x),y_i,\psi_i^{\sigma(n)}) = G(F_{\sigma(n)}(x), y_j, \psi_j^{\sigma(n)}) \leq G(x, y_k, \psi_k^{\sigma(n)}) \]
	Since $G$ is continuous the previous inequality passes to the limit  $n \to \infty$, showing that  $x \in L_\infty$, and that
	\[ \limsup_{n \to \infty} \chi_{L_n}(x) \leq \chi_{L_\infty}(x) \]
	We now want to show $\liminf_{n \to \infty} \chi_{L_n}(x) \geq \chi_{L_\infty}(x)$. If $x \notin L_\infty$ the result is straightforward. 
	Let us consider the set
	\begin{equation}
		\label{Sij}
		S_{ij} = \left( \bigcup_{k \neq i,j} \Gamma_{ijk}(\psi^\infty)  \right) \cup (\Gamma_{ij}(\psi^\infty) \cap \partial X)
	\end{equation}
	By the genericity hypothesis (Definition~\ref{Ygeneric}) we have $\mathcal{H}^{d-1}(S_{ij}) = 0$. If $x \in L_\infty \setminus S_{ij}$,  by definition we get for every $k\not\in\{i,j\}$ that $x$ does not belong to $\Gamma_{jk}(\psi^\infty)$. This implies a strict inequality
	\[ G(x, y_i, \psi_i^\infty) = G(x, y_j, \psi_j^\infty) < G(x, y_k, \psi_k^\infty). \]
	Since $F_n(x)$ converges to $x$ and since $\psi^n$ converges to $\psi^\infty$, we get for $n$ large enough 
	\[ 
	\begin{cases}
	G(F_n(x), y_i, \psi_i^n) < G(F_n(x), y_k, \psi_k^n) \\
	G(F_n(x), y_j, \psi_j^n) < G(F_n(x), y_k, \psi_k^n).
	\end{cases}
	\]
	Moreover since $x \in \Gamma_{ij}(\psi^\infty)$,  $F_n(x) \in \Gamma_{ij}(\psi^n)$. Combining the inequalities above, this shows that  $F_n(x)$ belongs to $\Lag_{ij}(\psi^n) \cap \widetilde{\Omega} = F_n(L_n)$, i.e. $x \in L_n$. This gives us
	\[
	\forall x\not\in S_{ij},\quad \liminf_{n \to \infty} \chi_{L_n}(x) \geq \chi_{L_\infty}(x).
	\]
        Consider $x\not\in S_{ij}$. For such $x$, we already know that $\chi_{L_n}(x) \to \chi_{L_\infty}(x)$ as $n\to+\infty$. Thus, if $x$ does not belong to $L_\infty$, we directly have
        $$ \lim_{n\to +\infty} \chi_{L_n}(x) \chi_X(F_n(x)) = \chi_{L_\infty}\chi_X(x) = 0.$$
        We may now assume that $x$ belongs to $L_\infty \setminus  S_{ij}$. By definition of $S_{ij}$, this implies that  $x \notin \partial X$. We  can directly deduce that $\chi_X$ is continuous at $x$ and that $\chi_X(F_n(x)) \to \chi_X(x)$ when $n \to +\infty$.
        
	In conclusion we have that $H_{ij}(\psi^n) \to H_{ij}(\psi^\infty)$, so that $H_{ij}$ is continuous.
\end{proof}
\subsection{Kernel and image of  $\D H$}\label{section:differential}
The goal of this section is to prove Proposition~\ref{prop:differentialproperties} that gives properties on the differential of the mass function $H$. 
We consider the admissible set
\begin{equation}
\label{B+}
\mathcal{S}^+ = \left\{ \psi \in \R^N \mid \forall i\in\OneN, H_i(\psi) > 0 \right\}.
\end{equation}
\begin{proposition}\label{prop:differentialproperties}
  \label{ImDH}
	In addition to the assumptions of  Proposition \ref{Hi_diff}, we assume that
        $$\inter(X) \cap \{ \rho > 0 \} \hbox{ is  path-connected }, $$
        where $\inter(X)$ is the interior of $X$. Then we have for any $\psi \in \mathcal{S}^+$
	\begin{itemize}
		\item The differential $DH(\psi)$ has rank $N-1$;
		\item The image of $DH$ is $\im( DH(\psi) ) = \1^{\perp}$ where $\1 = (1, \cdots, 1) \in \R^N$;
		\item For any $w \in \ker(DH(\psi)) \setminus \{ 0 \}$, we have for all $i \in \OneN, w_i \neq 0$ and all $w_i$ have the same sign.
	\end{itemize}	  
\end{proposition}
The next two lemmas have already been included in the recent survey on optimal transport involving the second and third authors~\cite{merigot:hal-02494446}, but we include them here for completeness. The proof of Proposition \ref{ImDH} is different from the previous work in optimal transport because $H$ is not symmetric.
\begin{lemma}
	\label{connectedness}
	Let $U \subset \R^d$ be a path-connected open set, and $S \subset \R^d$ be a closed
	set such that $\mathcal{H}^{d-1}(S) = 0$. 
	Then, $U \setminus S$ is path-connected.
\end{lemma}
\begin{proof}
	It suffices to treat the case where $U$ is an open ball, the general case will follow by standard connectedness arguments. Let $x, y \in U \setminus S$ be distinct points. 
	Since $U \setminus S$ is open, there exists $r > 0$ such that $\B(x, r)$ and $\B(y, r)$ are included in $U \setminus S$. 
	Consider the hyperplane  $H$ orthogonal to the
	segment $[x, y]$, and $\Pi_H$ the projection on $H$. 
	Then, since $\Pi_H$ is 1-Lipschitz,
	$\mathcal{H}^{d-1}(\Pi_H S) \leq \mathcal{H}^{d-1} (S) = 0$, so that $H \setminus \Pi_H S$ is dense in the hyperplane
	$H$. 
	In particular, there exists a point	$z \in \Pi_H (\B(x,r)) \setminus S = \Pi_H (\B(y,r)) \setminus S$.
	By construction the line $z + \R(y - x)$ avoids $S$ and passes through the balls
	$\B(x, r) \subset U \setminus S$ and $\B(y, r) \subset U \setminus S$. 
	This shows that the points $x, y$ can be connected in $U \setminus S$.
\end{proof}

We define for $\psi \in  \R^N$ the graph $\mathcal{G}_\psi = (V,E)$ with vertex set $V = \{1, \hdots, N\}$  with edges 
\[ E = \left\{(i,j) \in V^2 \mid  \frac{\partial H_i}{\partial \psi_j}(\psi) > 0 \right\} \]
We have the following result.
\begin{lemma}
\label{graph-connected}
	Under the assumptions of Proposition \ref{ImDH} and for $\psi \in \mathcal{S}^+$, the graph $\mathcal{G}_\psi$ is connected.
\end{lemma}

\begin{proof}
	Let $Z = \inter(X) \cap \{\rho > 0\}$, and $S = \bigcup_{ij}
        S_{ij}$ where $S_{ij}$ is defined in \eqref{Sij}.  From Lemma
        \ref{connectedness} the set $Z \setminus S$ is path connected,
        we also have $\mu (Z \setminus S) = 1$ since $\mu(\partial X)
        = \mu(S) = 0$.  Suppose that $\mathcal{G}_\psi$ is not
        connected.  Let $i_0 \in \OneN$, and let $I_0$ be the
        connected component of $i_0$ in the graph $\mathcal{G}_\psi$.
        We thus have $i_0 \in I_0 \neq  \OneN$.  We consider
        the two non-empty sets
	$$U_1 = \bigcup_{i \in I_0} \Lag_i(\psi) \cap (Z \setminus S) \text{ and } U_2 = \bigcup_{i \notin I_0} \Lag_i(\psi) \cap (Z \setminus S),$$
        which partition $Z\setminus S$ up to a Lebesgue-negligible set. Moreover,  since $\psi \in \mathcal{S}^+,$
	\begin{equation*}
	\begin{cases}
		U_1 \cup U_2 = Z \setminus S, \\
		0 < \mu (U_1) < 1, \\
		0 < \mu (U_2) < 1.
	\end{cases}
	\end{equation*}
	By construction $U_1$ and $U_2$ are closed sets in $Z \setminus S$. 
	Since $\mu (U_i) > 0$ we can pick $x$ and $y$ in $Z \setminus S$ such that $x \in U_1$ and $y \in U_2$. The $Z \setminus S$ being path-connected, we know that there exists a path $\gamma \in \mathcal{C}^0 ([0,1], Z \setminus S)$ satisfying $\gamma(0) = x$ and $\gamma(1) = y$.
	We let $t = \max \{s \in [0,1] | \gamma(s) \in U_1 \}$ and we are going to show that $\gamma(t) \in U_1 \cap U_2$. By construction, $\gamma(t)$ obviously belongs to $U_1$. 
	Now if $t = 1$ we have $\gamma(t) = y \in U_2$. If not, we have for all $\epsilon > 0$ that $\gamma(t + \epsilon)  \in U_2$. 
	Since $U_2$ is relatively closed in $Z\setminus$ and since $\gamma$ is continuous, we have $\gamma(t) \in U_2$. 
	Naming $z = \gamma(t)$, there  exists $i \in I_0$, $j \notin I_0$ such that $z \in \Lag_i(\psi) \cap \Lag_j(\psi)$. Moreover, since  $z \notin S$ we get that for any  $k \notin \{ i,j \}$,
        $$G(z, y_i, \psi_i) = G(z, y_j, \psi_j) > G(z, y_k, \psi_k).$$
	By continuity of $G$ we can deduce that there exists an open ball of radius $r > 0$ such that  
	\[\forall x \in \B(z,r), \forall k \notin \{i,j\}, G(x, y_i, \psi_i) > G(x, y_k, \psi_k) \]
	This implies that
	\[ \B(z,r) \cap \Gamma_{ij}(\psi) \subset \Lag_{ij}(\psi) \]
	where $\Gamma_{ij}(\psi)$ is defined in Definition \ref{Ygeneric}. By \eqref{Twist} condition and the inversion function theorem, we know that $\Gamma_{ij}(\psi)$ is a $d-1$ dimensional manifold and $z \in \Gamma_{ij}(\psi)$. 
	Moreover we have $\rho(z)  > 0$ because $z \in Z$ and $\rho$ is continuous on $Z \subset X$ by hypothesis. 
	We now have 
	\begin{align*}
		\frac{\partial H_i}{\partial \psi_j}(\psi) 
		&= \int_{\Lag_{ij}(\psi)} \rho (x)\frac{|\partial_vG(x,y_i,\psi_i )|}{\Vert\nabla_xG(x,y_j,\psi_j) - \nabla_xG(x,y_i,\psi_i )\Vert}d\mathcal{H}^{d-1}(x) \\
		&\geq \int_{\B(z,r) \cap \Gamma_{ij}(\psi)} \rho (x)\frac{|\partial_vG(x,y_i,\psi_i )|}{\Vert\nabla_xG(x,y_j,\psi_j) - \nabla_xG(x,y_i,\psi_i )\Vert}d\mathcal{H}^{d-1}(x) > 0
	\end{align*}
	which is a contradiction with the hypothesis that $i$ and $j$ are not connected in the graph
        $\mathcal{G}_\psi$.
\end{proof}


\begin{proof}[Proof of Proposition \ref{ImDH}]
	We note the matrix $M = DH(\psi)$, with coefficients $m_{i,j} = \partial H_i/\partial \psi_j(\psi)$. 
	We first show that $\ker(M^T) = \Span(\1)$. The inclusion $\1\in\ker(M^T)$ follows from
        \[ \sum_{i = 1}^N m_{i,j} = \frac{\partial }{\partial \psi_j}\left( \sum_{i = 1}^N H_i(\psi)  \right) = 0. \]
	Consider now  $v \in \ker(M^T)$, and pick an index $i_0$ where $v$ is maximum, i.e. $i_0 \in \argmax_{1 \leq i \leq N} v_i$. We have 
	$$0 = (M^T v)_{i_0} = \sum_{i=1}^N m_{i,i_0} v_i = \sum_{i\neq i_0} m_{i,i_0} v_i + m_{i_0,i_0}v_{i_0} = \sum_{i\neq i_0} m_{i,i_0} (v_i - v_{i_0}).$$
	Since $\psi \in \mathcal{S}^+$, we have by Proposition \ref{Hi_diff} that for $i \neq i_0$, $m_{i,i_0} \geq 0$. 
	By definition of $i_0$ we also have $v_i - v_{i_0} \leq 0$. 
	From all this we deduce that $v_i = v_{i_0}$ for any $i \neq i_0$ satisfying $m_{i,i_0} > 0$, i.e. any vertex $i$ adjacent to $i_0$ in the graph $\mathcal{G}_\psi$. By connectedness of $\mathcal{G}_\psi$, we  conclude that $v = v_{i_0} \1$, thus showing $\ker(M^T) = \Span(\1)$.
        
	We can deduce from this result that $M$ is of rank $N-1$ because $\rk(M) = \rk(M^T) = N-1$. 
	Moreover for any $u \in \R^N$,
	$$\langle \1, Mu \rangle = (Mu)^T \1 = u^T M^T \1 = 0.$$ Since
        the spaces $\im(M)$ and $\1^{\perp}$ have the same dimension,
          we immediately get $\im(M) = \1^{\perp}$.
	
	Let $w \in \ker(M)\setminus \{ 0 \}$, we now want to show that for all $i \in \llbracket 1,N \rrbracket, w_i \neq 0$ and that all of the $w_i$ have the same sign. 
	The proof consists in two steps:
	\begin{itemize}
		\item Step 1: we show that $w \geq 0$ (or $-w \geq 0$).
		\item Step 2: we show that for $i \in \llbracket 1,N \rrbracket$, $w_i > 0$.
	\end{itemize}
	We define $\lambda = \max_{i} |m_{i,i}|$ and  $A = \lambda I + M$. With these definitions, $v$ belongs to $\ker(M)$ if and only if $Av = \lambda v$. Moreover, for any
        $i,j \in \llbracket 1,N \rrbracket$, one has  $a_{i,j} \geq 0$ and
        $$\displaystyle \sum_{k=1}^N a_{k,j} = \lambda.$$
	
	\noindent \textbf{Step 1:} Assume that there exists $i_0 \in \llbracket 1,N \rrbracket$ such that $w_{i_0} \geq 0$ (we can do this without loss on generality, by working on $-w$ otherwise). 
	Suppose that there exists $j \neq i_0$ such that $a_{i_0,j} > 0$ and $w_j< 0$, then since $A w = \lambda w$, we have $\lambda w_{i_0} = \sum_{j=1}^N a_{i_0,j}w_j$ and thus $\lambda |w_{i_0}| < \sum_{j=1}^N a_{i,j}|w_j|$. 
	We also have for any $i \in \llbracket 1,N \rrbracket, \lambda |w_i| \leq \sum_{j=1}^N a_{i,j}|w_j|$. 
	By summing this inequality on $i$ and since the inequality is strict when $i=i_0$, we obtain
	$$ \sum_{i=1}^N \lambda |w_i| < \sum_{i=1}^N \sum_{j=1}^N a_{i,j}|w_j| =  \sum_{j=1}^N |w_j| \sum_{i=1}^N a_{i,j} = \sum_{j=1}^N \lambda |w_j|, $$
	which is a  contradiction, so we can affirm that there exists no index $j \neq i_0$ such that $w_j < 0$ and $a_{i_0,j} > 0$. Since $A = M + \lambda I$, for $j \neq i_0 , a_{i_0,j} = m_{i_0,j}$. We thus have $\forall j \in \OneN , m_{i_0,j} > 0 \implies w_j \geq 0$. 
	By connectedness of $\mathcal{G}$ we deduce $w \geq 0$.\\
	
	\noindent \textbf{Step 2:} 
	If there exists $i \in \llbracket 1,N \rrbracket$ such that $w_i = 0$, then $\sum_j a_{i,j}w_j = 0$. 
	Recall that by construction $a_{i,j} \geq 0$ and with step 1 $w_j \geq 0$, so we have $\forall j, a_{i,j} > 0 \implies w_j = 0$. 
	Again by connectedness of $\mathcal{G}$ we have $w = 0$.
	
\end{proof}

\begin{remark} Remark that a part of the proof of Proposition~\ref{ImDH} could also be seen as a consequence of the Perron Frobenius theorem, using the notions of irreducible and stochastic matrices. 
	The matrix $A = M + \lambda I$ can be written $A = \lambda S$ where $S^T$ is a stochastic matrix. 
	The matrix $S$ is thus of spectral radius $1$ and $A$ is of spectral radius $\lambda$. 
	Since $M$ is irreducible, $A$ is also irreducible.
	Perron Frobenius Theorem then implies that $\lambda$ is a simple eigenvalue with an associated eigenvector $w$ satisfying $w_i > 0$ for any $i \in \OneN$.
	Since $A v = \lambda v \iff M v = 0$, we can deduce that $\rk(M) = N-1$ and $\ker(M) = \Span(w)$.
	Moreover since $\1 \in \ker(M^T)$, we have for any $u \in \R^N$, $\langle \1 , M u \rangle = (M u)^T \1 = u^T M^T \1 = 0$ and $\im(M) = \1^\perp$.
\end{remark}

\subsection{Damped Newton algorithm}\label{sec:algo}
In this section, we present a damped Newton algorithm to solve the generated Jacobian equation \eqref{GenJac'}, namely $H(\psi)=\nu$. For this purpose we define in the following lemma an admissible set of variable that can be used in our algorithm.
\begin{lemma}[Admissible set]
	\label{Sadlemma}
	Suppose that the hypothesis of Proposition \ref{Hi_diff} are satisfied. For any $\delta > 0$, there exists $\alpha \in \R$ such that the set
	\begin{equation}
	\label{Sad}
	\Sad := \left\{\psi \in \R^N \mid \psi_1 = \alpha \text{ and } \forall i \in \llbracket 1,N \rrbracket , H_i(\psi) \geq \delta \right\}  \subset \mathcal{S}^+
	\end{equation}
	is a compact subset of $\R^N$. Furthermore for $\delta$ small enough, the set \eqref{Sad} is non-empty.
\end{lemma}
\begin{proof}
  Let $\gamma \in \R$ and $M = \max_{(x,y) \in X \times Y} G(x,y,\gamma)$, where $M$ is finite thanks to the continuity of $G$ and compactness of $X\times Y$. From the condition \eqref{Uniform Convergence}, there exists $\alpha\in\Rsp$ such that $\min_{x\in X} G(x,y_1,\alpha) > M$. If $\psi\in\Rsp^N$ is such that
  $\psi_1 = \alpha$ and $\psi_i > \gamma$ for some $i\geq 2$, then using \eqref{Monotonicity},
  $$\forall x\in X, G(x,y_1,\alpha) > M \geq G(x,y_i,\gamma) \geq
  G(x,y_i,\psi_i),$$ thus implying that $\Lag_{i}(\psi) = \emptyset$,
  and in particular $\psi\not\in\Sad$. We argue similarly to show an
  upper bound on the elements of $\Sad$: by \eqref{Uniform
    Convergence}, there exists $\beta \in \R$ such that $\min_{(x,y)
    \in X \times Y} G(x,y,\beta) > \max_{x \in X} G(x ,y_1, \alpha)$.
  If $\psi\in\Rsp^N$ is such that $\psi_1 = \alpha$ and $\psi_i <
  \beta$ for some $i\geq 2$, then using \eqref{Monotonicity}, we get
  $$ \forall x\in X, G(x,y_i,\psi_i) \geq G(x,y_i,\beta) >
  G(x,y_1,\alpha),$$ thus showing that $\Lag_1(\psi) = \emptyset$, so
  that $\psi\not\in \Sad$.  The set $\Sad$ can be written as $\Sad =
  \{ \alpha \} \times \cap H^{-1}( [\delta, 1]^N ),$ and is therefore
  closed by continuity of $H$. The previous computations show that
  $\Sad \subseteq \{\alpha\}\times [ \beta, \gamma ]^{N-1}$, proving
  that $\Sad$ is compact.

	Now suppose that $\delta \leq 1/2^{N-1}$, then we can iteratively construct a vector $\psi \in \Sad$ in the following way. We start from $\psi = (\alpha, \gamma, \cdots, \gamma) \in \R^N$. We then have $H_1(\psi) = 1$ and for any $i \geq 2, H_i(\psi) = 0$. Then for all $i$ from $2$ to $N$ can decrease $\psi_i$ such that $H_i(\psi) = 1 / 2^{i-1}$. Then after iteration $i$ we have
	\[ \forall k < i, H_k(\psi) \geq \frac{1}{2 ^{k-1}} - \sum_{k +1 \leq j \leq i} \frac{1}{2 ^{j-1}} = \frac{1}{2^{i-1}}   \]
	After iteration $N$ we thus have that for all $i \in \OneN, H_i(\psi) \geq 1/ 2^{N-1} \geq \delta$, and since $\psi_1 = \alpha$ has not been changed during the process we have $\psi \in \Sad$ and $\Sad \neq \emptyset$.
\end{proof}

The differential of $H$ is not invertible, but we can still define a Newton's direction by fixing one coordinate:
\begin{proposition}[Newton's direction]
	Under the assumptions of Proposition \ref{ImDH}, the system
	\begin{equation}
	\label{udes}
	\begin{cases}
	DH(\psi)u = H(\psi) - \nu\\
	u_1 = 0
	\end{cases}
	\end{equation}
	has a unique solution in $\R^N$.
\end{proposition}
\begin{proof}
Notice that from Proposition \ref{ImDH}, $DH(\psi)$ is of rank $N-1$ and since $H(\psi) - \nu \in \1^\perp = \im(DH(\psi))$, the set $S = \{u \in \R^N | DH(\psi)u = H(\psi) - \nu \}$ is of dimension 1. For $u \in S$ and $w \in \ker(DH(\psi))\setminus \{0\}$, $S = \{u + tw, t \in \R \}$. Since $w_1 \neq 0$ for $w \in \ker(DH(\psi)) \setminus \{0\}$, system \eqref{udes} has a unique solution.
\end{proof}
\begin{algorithm}
	\caption{Damped Newton algorithm to solve \eqref{GenJac'}}
	\label{newton}
	\begin{algorithmic}[1]
		\Require 
		$\epsilon > 0$; initialization $\psi^0 \in \Sad$ where $\delta \leq \min_i \nu_i /2$
		\Ensure $\psi$ such that $\Vert H(\psi)-\nu \Vert \leq \epsilon$
		\State $k \leftarrow 0$
		\While{$\Vert H(\psi^k) -\nu \Vert >  \epsilon$}
		\State Define  $u^k$ as the solution of the linear system
		$$
		\hspace{-5cm}
		\begin{cases}
		DH(\psi^k)u = H(\psi^k) - \nu\\
		u_1 = 0
		\end{cases}
		$$
		\State Compute $\tau^k$ by backtracking, i.e.
		\begin{align*}
		\displaystyle \tau^k = \max &\{ \tau \in 2^{-\N} | \psi^{k,{\tau}} = \psi^k - \tau u^k \in \Sad \text{ and } \\
		 &~\Vert H(\psi^{k,{\tau}}) - \nu\Vert \leq (1- \frac{\tau}{2})\Vert H(\psi^k) - \nu\Vert
		\end{align*}
		\State $\psi^{k+1} \leftarrow \psi^k - \tau^k u^k$ and $k \leftarrow k+1$
		\EndWhile
		\State \Return $\psi^k$
	\end{algorithmic}
\end{algorithm}
\begin{theorem}[Linear convergence]
  \label{th:convergence}
  Assume the following assumptions:
      \begin{itemize}
    \item the generating function $G \in \Class^2(\Omega\times Y \times \R)$ satisfies the assumptions
  \eqref{Regularity}, \eqref{Monotonicity}, \eqref{Twist}, \eqref{Uniform Convergence},
  \eqref{GenOmega}, \eqref{GenX},
  \item $X \subseteq \Omega$ is compact
    and  $\rho$ is a continuous probability density on $X$.
    \item $\inter(X) \cap \{ \rho > 0 \}$ is is  path-connected.
  \end{itemize}
      Then, there exists $\tau^*\in ]0,1]$ such that the iterates of Algorithm~\ref{newton} satisfy $$ \| H(\psi^k) - \nu \| \le \left(1 - \frac{\tau^*}{2}\right)^k \| H(\psi^0) - \nu \|.$$
          In particular, Algorithm~\ref{newton} terminates.
\end{theorem}
\begin{proof}
	Let $\psi^0 \in \Sad$, we define the set
	$$ K^{\delta} = \{ \psi \in \Sad , \Vert H(\psi)  - \nu \Vert \leq \Vert H(\psi_0) - \nu \Vert\}$$
	Since the function $H$ is continuous, the set $K^{\delta}$ is non-empty and compact.
	Note that system \eqref{udes} has $N+1$ lines for $N$ variables, and we know that the last line $u_1 = 0$, which can be written $e_1^T u = 0$, is linearly independent from the others. We can thus rewrite the system in the following form
	\begin{equation}
		\label{udes2}
		M(\psi) u = H(\psi) - \nu
	\end{equation}
	where $M(\psi) = DH(\psi) + e_1 e_1^T$. Obviously if $u$ is a solution of \eqref{udes} then it is also a solution of \eqref{udes2}. Now if $u$ is a solution of \eqref{udes2}, since $e_1 \notin \im(DH(\psi))$ and $H(\psi) - \nu \in \im(DH(\psi))$, we have $e_1 e_1^T u =  e_1^T u e_1= 0$ which means that the scalar $e_1^T u = 0$ and thus, u satisfies \eqref{udes}. Since \eqref{udes2} has a unique solution, $M(\psi)$ is thus invertible.
	Let $u_\psi$ solution of \eqref{udes2} for a given $\psi$. We have $u_\psi = M^{-1}(\psi) (H(\psi) - \nu )$. We thus have for any $\psi \in K^\delta$ that $\Vert u_\psi \Vert \leq \Vert M^{-1}(\psi)\Vert_{op}  \Vert (H(\psi) - \nu )\Vert$ where $ \Vert \cdot\Vert_{op}$ denotes the operator norm in $\mathcal{M}_N(\R)$. The function $\psi \mapsto M(\psi)$ is continuous and $M$ is invertible so $\psi \mapsto \Vert M^{-1}(\psi)\Vert_{op}$ is also continuous and admits a maximum on the compact set $K^\delta$. We note $C = \max_{\psi \in K^\delta} \Vert M^{-1}(\psi)\Vert_{op}$ so we have for any $\psi \in K^\delta$, $\Vert u_\psi\Vert \leq C \Vert H(\psi) - \nu\Vert$.\\
	Let $\psi \in K^\delta$ and $\psi^\tau = \psi - \tau u_\psi$ for $\tau \in [0,1]$. The first coordinate of $\psi^\tau$ satisfies $\psi_1^\tau = \alpha$. For a small $\tau$ we can write the Taylor expansion
	\begin{align*}
	H(\psi^\tau) &= H(\psi)  - \tau DH(\psi)u_\psi + o(\tau \Vert u_\psi\Vert)\\
	&= H(\psi) - \tau (H(\psi) - \nu) + o(\tau \Vert H(\psi) - \nu\Vert)
	\end{align*}
	it follows that
	$$\Vert H(\psi^\tau) - \nu \Vert = (1 - \tau)\Vert H(\psi) - \nu \Vert + o(\tau \Vert H(\psi) - \nu\Vert)$$
	and thus there exists $\tau^1_\psi > 0$ such that for all $\tau \in ]0, \tau^1_\psi[$
	$$\Vert H(\psi^\tau) - \nu \Vert \leq (1 - \frac{\tau}{2})\Vert H(\psi) - \nu \Vert$$
	By compactness of $K^\delta$, this property holds on an uniform open range $]0,\tau^1[$.\\
	Moreover, coordinatewise we have for $i \in \llbracket 1,N \rrbracket$, $$H_i(\psi^\tau) = (1 - \tau) H_i(\psi) + \tau \nu_i+ o(\tau \Vert H(\psi) - \nu\Vert)$$ and since $\nu_i\geq 2 \delta$ there exists $\tau^2_\psi > 0$ such that $$\forall \tau \in ]0,\tau^2_\psi[,~\forall i \in \OneN,~H_i(\psi^\tau) \geq (1 + \frac{\tau}{2}) \delta.$$ 	
	Then again by compactness of $K^\delta$, there exists $\tau^2 >0$ such that for all $\psi \in K^{\delta}$ and $\tau \in ]0, \tau^2[$, $\psi^\tau \in \Sad$. 
            This implies that the chosen $\tau^k$ in the algorithm will always be larger than
            $$ \tau^* = \frac{1}{2}\min(\tau^1,\tau^2).$$
            By definition of the iterates, we deduce at one that
            $$  \Vert H(\psi^{k+1}) - \nu \Vert \leq (1 - \frac{\tau^*}{2})\Vert H(\psi^k) - \nu \Vert, $$
            thus proving the desired convergence result.
\end{proof}
\begin{remark}[Existence] Note that the convergence of the algorithm allows to recover the existence of a
solution to the semi-discrete generated Jacobian equation. To obtain
this result we need the set $\Sad$ to be non-empty, which is the case by Lemma~\ref{Sadlemma} if $\delta$ is small enough. 
\end{remark}

\section{Application to the Near Field Parallel Reflector problem}\label{sec:application}
The Near Field Parallel Reflector problem is a non-imaging optics
problem that cannot be recast as an optimal transport
problem~\cite{kochengin1997determination,oliker2003mathematical}, but
that can be written as a generated Jacobian
equation~\cite{jiang2014pogorelov,abedin2017iterative}. We show in
this section that we can apply the Damped Newton algorithm to solve
this problem.

The description of the problem is as follows. We have a collimated light source (i.e. all the rays are parallel and vertical) emitted from an horizontal plane  $X \subset \R^2\times\{0\}$, whose intensity is modeled by a probability measure $\mu$ on $X$. We also have a target light, which is modeled by a probability measure $\nu$ supported on a finite set $Y \subset \R^2\times\{0\}$. 
The Near Field parallel reflector problem consists in finding the surface $\Sigma $ of a  mirror that reflects the measure $\mu$ to the measure $\nu$. 
Let us denote by $T_\Sigma : X \to Y$ the map that associates to any incident ray emanating from $x \in X$ the reflected direction $T_\Sigma(x)$ using Snell's law of reflection. 
The Near Field refractor problem then amounts to finding the mirror surface $\Sigma$ such that for any point $y \in Y$, 
\begin{equation}\label{eq:nearfield}\tag{NF paral}
\mu(T_\Sigma^{-1}(y)) = \nu(y).
\end{equation}

\begin{figure}[!h]
\begin{center}
\includegraphics[scale = 0.6]{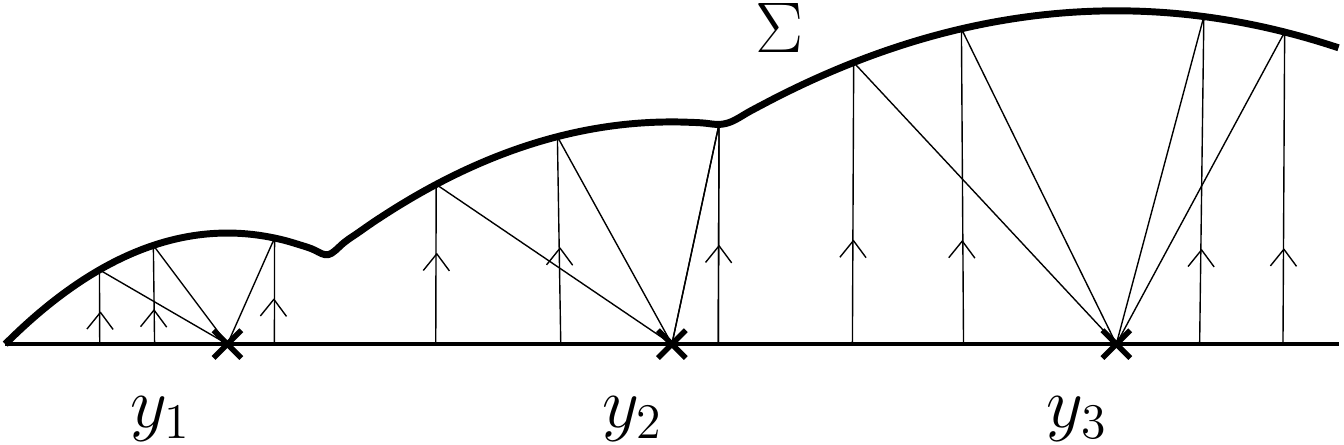}
\caption{A reflector composed of three paraboloids reflecting upward vertical rays toward the points $y_1, y_2, y_3$.}
\label{mirror}
\end{center}
\end{figure}

\subsection{Generated Jacobian equation.} 
In order to handle this problem, it is natural to consider paraboloid surfaces. Indeed, as illustrated in Figure~\ref{mirror}, a paraboloid $P(y,1 / \psi(y))$ with focal point $y\in \R^3$, focal distance $1 / \psi(y)$ and direction $(0,0,-1)$ reflects every upward vertical rays toward the focal point $y$. Since the target is finite, we choose to define the reflector surface $\Sigma$ as the upper envelop of a finite family of paraboloids $P(y,1 / \psi(y))$. Every paraboloid $P(y,1 / \psi(y))$ being the graph of the function $x\mapsto  1 / (2 \psi(y)) - \psi(y) \Vert x - y \Vert^2/2$ over $X$, the surface $\Sigma$ is the graph of the function 
 \[
 u(x) = \max_{y \in Y} \frac{1}{2 \psi(y)} - \frac{\psi(y)}{2} \Vert x - y \Vert^2 .
 \] 
We define  $G : \Omega \times Y \times \R_+^* \to \R$ by
\begin{equation}\label{eq:generatingfunctionoptic}
G(x,y,v) = \frac{1}{2 v} - \frac{v}{2} \Vert x - y \Vert^2
\end{equation}
where $\Omega$ is a bounded  open set containing $X$. Then for every $y\in Y$, one has $T_\Sigma^{-1}(y)=\Lag_y(\psi)$. 
In order to show that the semi-discrete version of Near Field
problem~\eqref{eq:nearfield} can be solved using our algorithm, we
need to show that the generating function $G$ satisfies all the
hypothesis of Definition~\ref{defG}.

The conditions $\eqref{Regularity}$, $\eqref{Monotonicity}$ and $\eqref{Uniform Convergence}$ are easy to verify, as mentioned in~\cite{abedin2017iterative}. This follows from the fact that $(x,y,v) \mapsto G(x,y,v)$ is continuously differentiable in $x$ and $v$, that $\nabla_x G(x,y,v) = v(y - x)$ and  that $\partial_v G(x,y,v) = -1/(2v^2) - v\Vert x - y\Vert^2/2$. The  $\eqref{Uniform Convergence}$ condition is satisfied because $\Omega$ is bounded. 
Concerning the Twist assumption, F. Abedin and C. Gutierrez~\cite{abedin2017iterative} introduce a necessary condition that they call \emph{Visibility condition}. This condition is that for any two point $y_i,y_j\in Y$ the line containing these two points does not intersect $X$. Since $X$ and $Y$ lie in the same plane $\R^2\times \{0\}$, this condition is quite restrictive in practice. We show below that it is not necessary here, since it is sufficient to have the \eqref{Twist} Condition on some interval $]0, \gamma[$ with $\gamma \in \R_+$.
\begin{proposition} 
	The function $G$ satisfies the \eqref{Twist} condition on $X\times Y\times ]0,\gamma[$ where $\gamma$ satisfies 
	\[ 
	\gamma <  \inf_{(x,y) \in X \times Y} \frac{1}{\| x - y \|}
	\]
\end{proposition}
\begin{proof}
	Let $x \in X$,  and suppose that $G(x,y_1,v_1)=G(x,y_2,v_2)$ and that $\nabla_x G(x,y_1,v_1)=\nabla_x G(x,y_2,v_2)$, with $v_i\in ]0,\gamma]$ and $y_i\in Y$. 
	The second condition implies that $v_1(y_1-x)=v_2(y_2-x)$, which implies that $x$, $y_1$ and $y_2$ are collinear. 
	We then have $y_1-x=(v_2/v_1) (y_2-x)$. 
	Plugging this in the relation $G(x,y_1,v_1)=G(x,y_2,v_2)$ gives
	\[
	\frac{1}{2v_1} - \frac{v_1}{2}\frac{v_2^2}{v_1^2}\|x- y_2\|^2 = \frac{1}{2v_2} - \frac{v_2}{2}\|x- y_2\|^2
	\]
	which gives
	\[
	\frac{1}{2v_1}( 1- v_2^2\|x-y_2\|^2) = \frac{1}{2v_2}( 1- v_2^2\|x-y_2\|^2),
	\]
	thus we have either $(y_1, v_1) = (y_2, v_2)$ or $v_2=1/\|x-y_2\|$. 
	The latter implying that $v_2 > \gamma$, which is not possible since by assumption $v_2 \leq \gamma$.
	It follows that $y,v \mapsto (G(x,y,v), \nabla_x G(x,y,v))$ is injective on $Y \times ]0, \gamma]$ for any $x \in X$.
\end{proof}

\subsection{Laguerre and M\"obius diagram.} 
In order to solve the Generated Jacobian equation~\eqref{eq:nearfield} with the Damped Newton algorithm, we study the Laguerre diagram induced by the generating function $G$. We observe that it is a particular instance of a M\"obius diagram~\cite{boissonat2007}. This will be useful to get a geometric condition that implies genericity (necessary to apply Algorithm~\ref{newton}) and it will also be used for the numerical computation of the Laguerre diagram. 

\begin{definition}[M\"obius diagram]
The M\"obius diagram of a family $\omega = (\omega_i)_{1 \leq i \leq N}$ of $N$ triplets $\omega_i = (\lambda_i, \mu_i, p_i) \in \R \times \R \times \R^d$ is the decomposition of the space into M\"obius cells $M_i(\omega)$ defined by
\[ M_i(\omega) = \left\{ x \in \R^d | \forall j \in \OneN,  \lambda_i \|x - p_i \|^2 - \mu_i \leq \lambda_j \|x - p_j \|^2 - \mu_j  \right\}\]
\end{definition}
A simple calculation shows the boundary of M\"obius cells is composed of arc of (possibly degenerated) circles~\cite{boissonat2007}.
\begin{proposition}
For any $p_i \neq p_j$, the intersection $M_i(\omega) \cap M_j(\omega)$ between two M\"obius cells is either empty, or an arc of circle whose center belong to the line passing through $p_i$ and $p_j$, or the bisector of $p_i$ and $p_j$.
\end{proposition}

Note that if we define $\lambda_i = \psi_i / 2$, $\mu_i = 1 / 2 \psi_i$ and $p_i = y_i$, then the Laguerre cells are M\"obius cells, namely
$$\Lag_i(\psi) = M_i(\omega) \cap \Omega.$$
This allows to show that the conditions \eqref{GenOmega} and \eqref{GenX} that are required to show the convergence of Algorithm~\ref{newton} are not restrictive. Indeed, by the previous proposition, the interface $\Gamma_{ij}(\psi)$ between the two Laguerre cells associated to $y_i$ and $y_j$ is contained in a circle for which the center is on the line passing through $y_i$ and $y_j$.
This circle can degenerate into a line, in this case it is the bisector between $y_i$ and $y_j$.
Suppose that $Y$ does not contain three colinear points, then for any distinct $i,j,k$, $\Gamma_{ijk}(\psi)$ is the intersection of two circles with different centers and \eqref{GenOmega} is satisfied.
Similarly if $\partial X$ doesn't contain any circle arc, nor bisectors of any two points of $Y$, then \eqref{GenX} is also satisfied. This allows to prove the following theorem.

\begin{theorem} 
Suppose that $Y$ does not contain three aligned points, and that $\partial X$ doesn't contain any circle arc, nor bisectors of any two points of $Y$. Assuming that the measures $\mu$ and $\nu$ satisfy the mass balance $\mu(X) = \nu(Y)$ and that $\mu$ is absolutely continuous with a continuous density $\rho$ such that $\inter(X) \cap \{ \rho > 0 \}$ is path-connected. Then the Damped newton Algorithm (Algorithm~\ref{newton}) converges toward a solution of \eqref{eq:nearfield}.
\end{theorem}

\begin{remark}
The Generating function is defined on $\Omega \times Y \times ]0, \gamma[$ instead of $\Omega \times Y \times \R$. As mentioned in Remark \ref{interval}, if $\zeta : \R \to ]0, \gamma[$ is a $\Class^1$-diffeomorphism, then the function $\widetilde{G}$ defined by $\widetilde{G}(x,y,v) = G(x,y,\zeta(v))$ is a generating function defined on $\Omega \times Y \times \R$ and we can apply Algorithm~\ref{newton} to $\widetilde{G}$.
\end{remark}

\subsection{Implementation.} The main difficulty in the implementation of the Newton algorithm is the evaluation of the function $H$ and of its differential $DH$, which requires an accurate computation of the Laguerre diagram. For this, we use the fact that a M\"obius diagram can be obtained by intersecting a 3D Power diagram with a paraboloid~\cite{boissonat2007}.

\begin{definition}[Power diagram]
The power diagram of a set of $N$ weighted points $\mathcal{P} =((p_i,r_i))_{1\leq i \leq N}$ where $p_i \in \Rsp^d$ and $r_i \in \Rsp$ is the decomposition of the space into  Power cells given by 
\[ 
\Pow_i(\mathcal{P}) = \{ x \in \R^d |  \forall j \in \OneN: \Vert x - p_i\Vert^2 - r_i \leq \Vert x - p_j\Vert^2 - r_j \} 
\]
\end{definition}
\begin{proposition}
\label{mobiuspower}The Laguerre cells associated to the generating function $G$ defined in \eqref{eq:generatingfunctionoptic} are given for any $i$ by 
\[
\Lag_i(\psi)=  \Pi(\Pow_i(\mathcal{P}) \cap P) \cap \Omega, 
\] 
where $P$ is the paraboloid in $\R^3$ parametrized by $ x_3 = x_1^2 + x_2^2$,  $\Pi$ is the projection of $\R^3$ on $\R^2$  defined by $\Pi(x,y,z) = (x,y)$, and $(\Pow_i(\mathcal{P}))_{1 \leq i \leq N}$ is the Power diagram associated to the weighted points $\mathcal{P}$ given by
\begin{equation}
\label{ciri}
\forall i \in \OneN : 
\begin{cases}
\displaystyle p_i = \left(\frac{\psi_i}{2}y_i, \frac{-\psi_i}{4 } \right) \\
\displaystyle r_i = \frac{\psi_i^2}{16} + \frac{\psi_i^2 \Vert y_i\Vert^2}{4} - \frac{\psi_i \Vert y_i\Vert^2}{2} + \frac{1}{2 \psi_i}
\end{cases}
\end{equation}
\end{proposition}

In our implementation of the algorithm, the intersection of power diagrams with a paraboloid is computed using an algorithm presented in \cite{machadomanhaesdecastro:hal-00952720}. Once the diagram is computed, the function $H$ and its differential $DH$ are computed using the trapezoidal rule.
Numerical experiments are performed with $X = [-1, 1]^2$ and $\mu$ equal to (one fourth) of the restriction of the  Lebesgue measure on $X$. The set $Y$ is randomly generated in the square $[0,1]^2$ for different values of $N$, associated with a discrete uniform measure $\nu$. Figure~\ref{initialdiagram} (left) shows the initial diagram $(\Lag_i(\psi))_{1 \leq i \leq N}$ with $N = 5000$ for some vector $\psi = \lambda \1$ with $\lambda > 0$. Figure~\ref{initialdiagram} (right) is the same diagram after convergence of the algorithm, where $\psi$ is an approximate solution of \eqref{eq:nearfield}. The graph of Figure~\ref{errorgraph} represents the error $\|H(\psi^k) - \nu\|_1$ as a function of iteration $k$. It shows superlinear convergence of the damped Newton method.
\begin{figure}
\begin{center}
  \includegraphics[width=.48\textwidth]{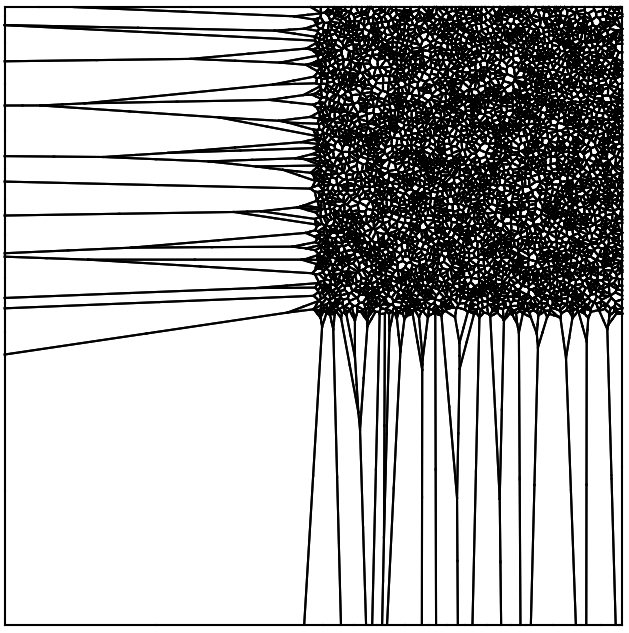}
  \includegraphics[width=.48\textwidth]{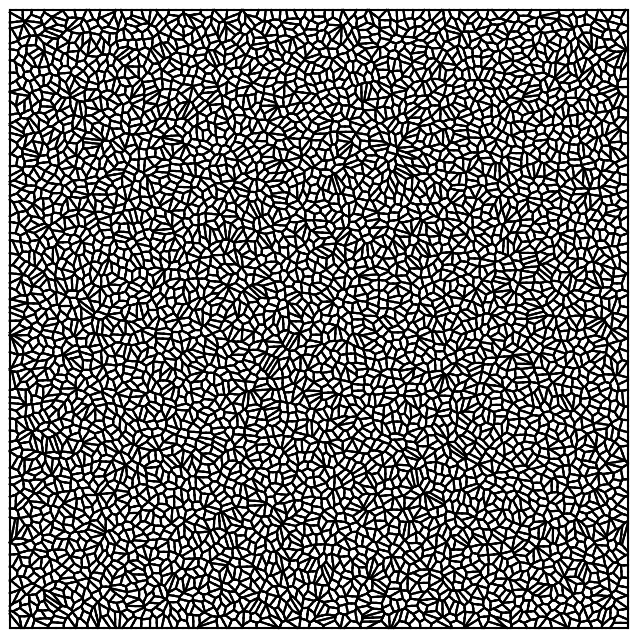}
\caption{Initial diagram for $N = 5000$, and final diagram, after convergence of the algorithm.}
\label{initialdiagram}
\end{center}
\end{figure}

\begin{figure}
\begin{center}
\includegraphics[scale = 0.7]{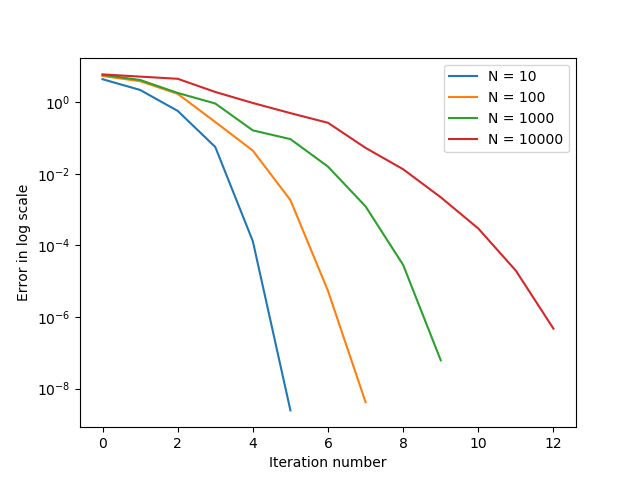}
\caption{Numerical error $\|H(\psi^k) - \nu\|_1$ as a function of the iteration $k$.}
\label{errorgraph}
\end{center}
\end{figure}

\subsection*{Acknowledgements}
We acknowledge the support of the French Agence Nationale de la
Recherche through the project MAGA (ANR-16-CE40-0014).

\bibliographystyle{plain}

\end{document}